\documentclass[11pt]{article} 

\usepackage[utf8]{inputenc} 

\usepackage{geometry} 
\usepackage{graphicx,psfrag} 
\usepackage{color}

\definecolor{Red}{rgb}{1,0,0}
\definecolor{Blue}{rgb}{0,0,1}
\definecolor{Olive}{rgb}{0.41,0.55,0.13}
\definecolor{Green}{rgb}{0,1,0}
\definecolor{MGreen}{rgb}{0,0.8,0}
\definecolor{DGreen}{rgb}{0,0.55,0}
\definecolor{Yellow}{rgb}{1,1,0}
\definecolor{Cyan}{rgb}{0,1,1}
\definecolor{Magenta}{rgb}{1,0,1}
\definecolor{Orange}{rgb}{1,.5,0}
\definecolor{Violet}{rgb}{.5,0,.5}
\definecolor{Purple}{rgb}{.75,0,.25}
\definecolor{Brown}{rgb}{.75,.5,.25}
\definecolor{Grey}{rgb}{.5,.5,.5}

\usepackage{booktabs} 
\usepackage{array} 
\usepackage{paralist} 
\usepackage{verbatim} 
\usepackage{subfigure} 
\usepackage{amsmath,amssymb,amsthm}

\usepackage{fancyhdr} 
\usepackage{sectsty}
\allsectionsfont{\sffamily\mdseries\upshape} 

\theoremstyle{plain}

\newtheorem{corollary}{Corollary}
\newtheorem{claim}{Claim}

\newtheorem{conj}{Conjecture}
\newtheorem{bound}{Bound}
\newtheorem{fact}{Fact}

\theoremstyle{remark}

\newtheorem{remark}{Remark}
\newtheorem{observation}{Observation}

\theoremstyle{definition}

\newcommand{\p}{{\rm P}}
\newcommand{\e}{{\rm E}}

\def\cU{{\cal U}}
\def\cV{{\cal V}}
\def\cW{{\cal W}}
\def\cX{{\cal X}}

\title{An information inequality for the BSSC channel}
\author{ Varun  Jog \and Chandra Nair}

\begin{document}
\maketitle
\begin{abstract}
We establish an information theoretic inequality concerning the binary skew-symmetric broadcast channel that was conjectured by one of the authors. This inequality helps to quantify the gap between the sum rate obtained by the inner bound  and outer bound for the binary skew-symmetric broadcast channel. 
\end{abstract}

\section{Introduction}

The broadcast channel is a fundamental network information theory setting modeling the communication of messages (private and common) from a single sender to multiple receivers. For formal definitions and early prior work the reader is referred to \cite{cov72,cov98}. There has been some recent progress for the discrete memoryless setting, and this work establishes a conjecture proposed in one of the recent papers \cite{naw08}.

We consider the broadcast channel where sender $X$ wishes to communicate independent messages $M_{1},M_{2}$ to two receivers $Y_{1},Y_{2}$. The capacity region for the broadcast channel is an open problem and the best known achievable region is due to Marton\cite{mar79} and is presented below.

\begin{bound}
\label{bd:inner}
\cite{mar79} The following region is achievable
\begin{align*}
R_1 &\leq I(U,W;Y_1) \\
R_2 &\leq I(V,W;Y_2) \\
R_1 + R_2 & \leq I(U,W;Y_1) + I(V;Y_2|W) - I(U;V|W) \\
R_1 + R_2 & \leq I(V,W;Y_2) + I(U;Y_1|W) - I(U;V|W) 
\end{align*}
for any triple of random variables $p(u,v,w)$ such that $(U,V,W) \to X \to (Y_1,Y_2)$ form a Markov chain.
\end{bound}

Capacity regions have been established for a number of special cases and in \emph{every case} where capacity is known, the following outer bound and Marton's inner bound yields the same region.

\begin{bound}
\label{bd:outer}
\cite{nae07} The union of rate pairs
\begin{align*}
R_1 &\leq I(U;Y_1) \\
R_2 &\leq I(V;Y_2) \\
R_1 + R_2 & \leq I(U;Y_1) + I(V;Y_2|U) \\
R_1 + R_2 & \leq I(V;Y_2) + I(U;Y_1|V) 
\end{align*}
over pairs of random variables $p(u,v)$   such that $(U,V) \to X \to (Y_1,Y_2)$ form a Markov chain constitutes an outer bound to the capacity region.
\end{bound}

\begin{figure}[ht]
\begin{center}
\begin{psfrags}
\psfrag{X}[r]{$X$}
\psfrag{Y}[l]{$Y_1$}
\psfrag{Z}[l]{$Y_2$}
\psfrag{p}[b]{$\frac 12$}
\psfrag{1-p}[c]{$\frac 12$}
\psfrag{0}[c]{$0$}
\psfrag{1}[c]{$1$}
\includegraphics[width=0.4\linewidth,angle=0]{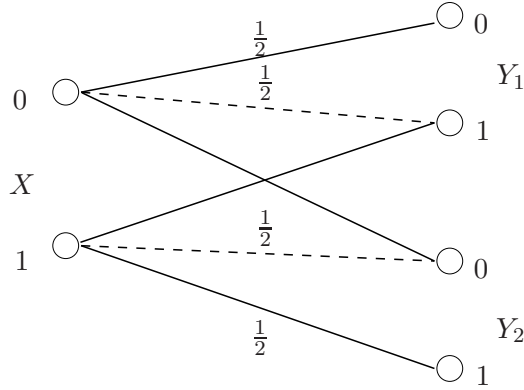}
\end{psfrags}
\caption{binary skew-symmetric broadcast channel}
\label{fig:bssc}
\end{center}
\end{figure}

In \cite{nae07} the authors studied Bound 2 for the binary skew symmetric channel  and showed that the line segment\footnote{There is a typo in the evaluation of this bound in the original paper though the main results are right. In the Appendix we will show the corrected derivation of this bound.} of $R_{1}+R_{2}=0.3725..$ lies on the boundary of the outer bound. In \cite{naw08} the authors studied Marton's inner bound for the binary skew-symmetric broadcast channel and showed that provided an information theoretic inequality (Conjecture \ref{conj:intro}) holds, a line segment of $R_{1}+R_{2}=0.3616...$ lies on the boundary of the Marton's inner bound.

\begin{conj}
\label{conj:intro} \cite{naw08} For the binary skew-symmetric channel shown in Figure \ref{fig:bssc},
$$ I(U;Y_{1})+I(V;Y_{2})-I(U;V) \leq \mbox{max}(I(X,Y_{1}),I(X,Y_{2}))$$
for all $(U,V,X)$  such that $(U,V)-X-(Y_{1},Y_{2})$ forms a Markov chain.
\end{conj}
It should be noted that this inequality was established in \cite{hap79} when $U,V$ were independent; and in \cite{naw08} for dependent $U,V$ and $P(X=0) \in [0,\frac{1}{5}] \cup [\frac{4}{5},1]$.

The outline of the proof is as follows: (Parts 1 and 2 were established in \cite{goa09} and is presented for completeness)
\begin{enumerate}
\item 
We show that to obtain the maximum sum-rate in Marton's region it is sufficient to consider $X$ as a function of the auxiliary random variables $U$ and $V$.
\item 
We further show that the cardinality of $U$ and $V$ can be restricted to $|\cX|$, which in this case is 2.
\item 
As $U$ and $V$ are binary, we conclude that the conjecture is true iff it holds true for each of the 16 cases where $X=f(U,V)$. We further prove that it is sufficient to consider only the two cases $X=U \land V$ and $X=U \oplus V$ by showing that either the other cases are immediate or it reduces to the above two cases.
\item
We prove the conjecture in $X=U \land V$ case using a multiplicative perturbation to derive properties about the distribution on $U,V$ which achieves the maxima in LHS of \ref{conj:intro}. 
\item
Similarly, we prove the conjecture in $X=U \oplus V$ case using an additive perturbation to derive properties about the distribution on $U,V$ which achieves the maxima in LHS of \ref{conj:intro}.

\end{enumerate}
\section{Main}
\label{sec;main}


Gohari and Anantharam \cite{goa09} established bounds on the cardinalities of the auxiliary random variables needed to evaluate the Marton's achievable region. 
In this section, we present a modified version of their arguments for completeness.

Define the following three quantities:
\begin{itemize}
\item $M = \max( I(U;Y_{1})+I(V;Y_{2})-I(U;V))$ over all $(U,V,X)$ such that  $(U,V)-X-(Y_{1},Y_{2})$ forms a Markov chain.
\item $M_d = \max (I(U;Y_{1})+I(V;Y_{2})-I(U;V))$ over all $(U,V,X)$ such that  $(U,V)-X-(Y_{1},Y_{2})$ forms a Markov chain and $X=f(U,V)$.
\item $M_d^{(|\cX|)} = \max (I(U;Y_{1})+I(V;Y_{2})-I(U;V))$ over all $(U,V,X)$ such that  $(U,V)-X-(Y_{1},Y_{2})$ forms a Markov chain;  $|\cU| \leq |\cX|$, $|\cV| \leq |\cX|$,
 and $X=f(U,V)$.
\end{itemize}

We will show that $M= M_d = M_d^{(|\cX|)}$ for any discrete-memoryless broadcast channel.

\begin{fact}
\label{fa:det}
$M=M_d$
\end{fact}
\begin{proof}
Using standard arguments, for e.g. \cite{hap79}, there exists a random variable $W$ independent of $U,V$ such that $X=f(U,V,W)$. Now set $V'=(V,W)$ and observe that
$I(U;Y_{1})+I(V;Y_{2})-I(U;V) \leq I(U;Y_{1})+I(V';Y_{2})-I(U;V')$.
\end{proof}

\begin{remark}
One way to construct such a $W$ is the following\footnote{This construction was mentioned to one of the authors by Bruce Hajek.}: For every $u,v$ consider the sequence $t_i(u,v) = \p(X \leq i|U=u,V=v), 1 \leq i \leq |\cX|$. Mark the points $t_1(u,v), \ldots, t_{|\cX|}(u,v)(=1)$ for all choices of $(u,v)$ along the unit interval $[0,1]$. The points define intervals (at most $|\cU| |\cV| |\cX|$) and generate $W$ as an independent random variable with probabilities defined by the length of the intervals. As the $\p(X = i|U=u,V=v)$ can be thought of as $W$ falling in a certain consecutive set of appropriately chosen intervals, there is a natural mapping  $(U,V,W) \mapsto X$.
\end{remark}

\begin{claim}
\label{cl:card}
$M_d = M_d^{(|\cX|)}$
\end{claim}
\begin{proof}
This is a simplified version of the arguments of Gohari and Anantharam \cite{goa09}, adapted to this setting. For a given $p(u,v,x)$ consider the multiplicative Lyapunov perturbation defined by
$ q(u,v,x) = p(u,v,x)(1 + \epsilon L(u))$. For $q(u,v,x)$ to be a valid probability distribution we require the following two conditions: $1 + \epsilon L(u) \geq 0, \forall u$ and $\sum_u p(u)L(u) = 0$.  {\it Note:} If $p(u,v,x)=0$ then $q(u,v,x)=0$ and hence $X$ continues to be a function of $(U,V)$ under any such perturbation.

If distribution $p(u,v,x)$ maximizes $I(U;Y_{1})+I(V;Y_{2})-I(U;V)$ then we must have that for any valid perturbation
\begin{enumerate}
\item $\frac{\partial}{\partial \epsilon} I(U;Y_{1})+I(V;Y_{2})-I(U;V) = 0,$
\item $\frac{\partial^2}{\partial \epsilon^2} I(U;Y_{1})+I(V;Y_{2})-I(U;V) \leq 0.$
\end{enumerate}

Consider a class of perturbations $L(u)$ such that 
\begin{equation}
\label{eq:cons}
\e(L|X=x)=\sum_{u,v} p(u,v|x) L(u) = 0, \forall x \in \cX.
\end{equation}
 Observe that these perturbations keep the distributions of $\cX$ (hence $Y_1,Y_2$) unchanged.

\begin{observation}
There exists such a non-zero perturbation if $|\cU| > |\cX|$ since the null-space of the constraints have rank at most $|\cX|$.
\end{observation}

Observe that
\begin{align*}
& I_q(U;Y_{1})+I_q(V;Y_{2})-I_q(U;V) \\
&\quad = H_q(Y_1) + H_q(Y_2) + H_q(U,V) - H_q(U,Y_1) - H_q(V,Y_2) \\
& \quad = H_p(Y_1) + H_p(Y_2) + H_q(U,V) - H_q(U,Y_1) - H_q(V,Y_2) \\
& \quad = H_p(Y_1) + H_p(Y_2) + H_p(U,V) + \epsilon H_p^L(U,V) - H_p(U,Y_1) - \epsilon H_p^L(U,Y_1) - H_q(V,Y_2).
\end{align*}
Here $H_p^L(U,V) = -\sum_{u,v} p(u,v)L(u)\log p(u,v)$, $H_p^L(U,Y_1) = -\sum_{u,y_1} p(u,y_1)L(u)\log p(u,y_1)$.

Therefore, $\frac{\partial^2}{\partial \epsilon^2} I(U;Y_{1})+I(V;Y_{2})-I(U;V) \leq 0$ implies $\frac{\partial^2}{\partial \epsilon^2} H_q(V,Y_2) \geq 0$ and this implies
$$ E(E(L|V,Y_2)^2) \leq 0$$
or in particular $E(L|V,Y_2) = 0$ whenever $p(v,y_2) \neq 0$. This, in turn, implies
$$ H_q(V,Y_2) = H_p(V,Y_2).$$

Using this we obtain
\begin{align*}
& I_q(U;Y_{1})+I_q(V;Y_{2})-I_q(U;V) \\
& \quad = H_p(Y_1) + H_p(Y_2) + H_p(U,V) + \epsilon H_p^L(U,V) - H_p(U,Y_1) - \epsilon H_p^L(U,Y_1) - H_q(V,Y_2) \\
& \quad = H_p(Y_1) + H_p(Y_2) + H_p(U,V) + \epsilon H_p^L(U,V) - H_p(U,Y_1) - \epsilon H_p^L(U,Y_1) - H_p(V,Y_2).
\end{align*}

The first derivative being zero implies $H_p^L(U,V) - H_p^L(U,Y_1) = 0 $ and this finally implies that if $p(u,v,x)$ attains the maximum of $I(U;Y_{1})+I(V;Y_{2})-I(U;V)$
then $I_q(U;Y_{1})+I_q(V;Y_{2})-I_q(U;V) = I_p(U;Y_{1})+I_p(V;Y_{2})-I_p(U;V)$ for any valid perturbation that satisfies $\eqref{eq:cons}$.

Now we choose $\epsilon$ such that $\min_u 1+\epsilon L(u) = 0$, and let $u=u^*$ achieve this minimum. Observe that $q(u^*) = 0$ and hence there exists an $U$ with cardinality equal to $|\cU|-1$ (at most) such that $I(U;Y_{1})+I(V;Y_{2})-I(U;V)$ is constant. We can proceed by induction until $|\cU| = |\cX|$. Observe that when $|\cU| = |\cX|$, we are no longer guaranteed the existence of a non-trivial $L(u)$ satisfying $\eqref{eq:cons}$.

The argument can then be repeated for $V$ to make $|\cV| \leq |\cX|$ as well.

This completes the proof that $M_d = M_d^{(|\cX|)}$.
\end{proof}

\begin{remark}
\label{rem:main}
Use Fact \ref{fa:det} and Claim \ref{cl:card}, to prove the conjecture \ref{conj:intro} it suffices to consider binary $U,V$ and $X=f(U,V)$. There are $16$ possible boolean functions on binary $(U,V)$ and we establish the conjecture for each such function.
\end{remark}

We use the following notation:  $U \land V$ (and),  $U \lor V$ (or), $U \oplus V$ (xor), $\bar{U}$ (not).

\begin{observation}
\label{ob:key}
Each of the  following groups of functions are equivalent upto re-labeling (of either $U$ or $V$ or both)
\begin{itemize}
\item $X=U, X = \bar U$,
\item $X=V, X = \bar V$,
\item $X=U \land V$, $X=\bar{U} \land V$, $X=U \land \bar{V}$, $X=\bar{U} \land \bar{V}$,
\item $X=U \lor V$, $X=\bar{U} \lor V$, $X=U \lor \bar{V}$, $X=\bar{U} \lor \bar{V}$,
\item $X=U \oplus V$, $X=\bar{U} \oplus V$
\end{itemize}
\end{observation}

\begin{claim}
\label{cl:triv}
The conjecture is valid when $X = 0, X = 1, X = U, X = V$.
\end{claim}
\begin{proof}
In the first two cases, the conjecture reduces to $-I(U;V) \leq 0$ (true by non-negativity of mutual information). In the third case conjecture follows from data processing inequality as $I(V;Y_2) \leq  I(V;X)$ and hence $I(X;Y_1) + I(V;Y_2) - I(V;X) \leq I(X;Y_1)$. The fourth case follows in a similar manner as the third.
\end{proof}

\begin{claim}
\label{cl:cond}
The conjecture is valid for all distributions $p(u,v)$ when $X=U\land V$ , {\em if and only if} the conjecture is valid for all distributions $q(u,v)$ when $X=U \lor V$.
\end{claim}
\begin{proof}
This follows from the skew-symmetry of the channel and that $X = U \lor V$ is equivalent to $\bar X = \bar U \land \bar V$. Let $P(U=i,V=j)=p_{ij}$ for every $i,j \in \{0,1\}$. For concreteness, when $X= U \land V$ the conjecture is equivalent to
\begin{align}
& h(\frac{p_{00} + p_{01} + p_{10}}{2}) - (p_{00}+p_{01})h(\frac 12) - (p_{10}+p_{11})h(\frac{p_{10}}{2(p_{10}+p_{11})}) + h(\frac{p_{11}}{2}) \nonumber\\
&  - (p_{01}+p_{11})h(\frac{p_{11}}{2(p_{01}+p_{11})}) - h(p_{00}+p_{01}) + (p_{00}+p_{10})h(\frac{p_{00}}{p_{00}+p_{10}}) + (p_{01}+p_{11})h(\frac{p_{01}}{p_{01}+p_{11}}) \nonumber \\
&\quad \leq \max \left\{ h(\frac{p_{00} + p_{01} + p_{10}}{2}) - (p_{00} + p_{01} + p_{10}) h(\frac 12), h(\frac{p_{11}}{2}) - p_{11} h(\frac 12) \right\},
\label{eq:and}
\end{align}
where $h(x) = -x \log_2 (x) - (1-x) \log_2(1-x)$ represents the binary entropy function.

\medskip
For the $X= U \lor V$ case, let $P(U=i,V=j)=q_{ij}$ for every $i,j \in {0,1}$.\\
The conjecture is now equivalent to
\begin{align}
& h(\frac{q_{00}}{2}) - (q_{01}+q_{00})h(\frac{q_{00}}{2(q_{01}+q_{00})}) +  h(\frac{q_{11} + q_{01} + q_{10}}{2}) - (q_{11}+q_{01})h(\frac 12) \nonumber \\
&    - (q_{10}+q_{00})h(\frac{q_{10}}{2(q_{10}+q_{00})})   - h(q_{11}+q_{01}) + (q_{11}+q_{10})h(\frac{q_{11}}{q_{11}+q_{10}}) + (q_{01}+q_{00})h(\frac{q_{01}}{q_{01}+q_{00}}) \nonumber \\
&\quad \leq \max \left\{ h(\frac{q_{11} + q_{01} + q_{10}}{2}) - (q_{11} + q_{01} + q_{10}) h(\frac 12), h(\frac{q_{00}}{2}) - q_{00} h(\frac 12) \right\}, \label{eq:or}
\end{align}

The bijection $p_{00} \leftrightarrow q_{11}, p_{01} \leftrightarrow q_{01}, p_{10} \leftrightarrow q_{10}, p_{11} \leftrightarrow q_{00}$ completes the proof of the equivalence of the conjectures under the constraints $X= U \land V$  and $X= U \lor V$.
\end{proof}

\begin{corollary}
From Remark \ref{rem:main}, Observation \ref{ob:key}, Claims \ref{cl:triv}, and \ref{cl:cond} it follows that the conjecture is true provided it holds when $X=U \land V$ and $X = U \oplus V$.
\end{corollary}

\subsection{Case 1: X= $U \wedge V$}

We prove the conjecture in this case by studying the local maxima. Clearly the conjecture is true when two of the three terms $p_{00},p_{01},p_{10}$ are identically $0$. When this happens, then the condition reduces to $X=U,V=1$, $X=V,U=1$, or $U=V=X$, each of which is solved by Claim \ref{cl:triv}. Clearly if $p_{11}=0$ then $X=0$; in which case the conjecture is valid. So we assume that $p_{11} > 0$. Therefore, we only establish the validity of the conjecture for the remaining cases.

Consider a perturbation $q(u,v,x)=p(u,v,x)(1+ \epsilon L(u,v))$ that maintains $\p(X=0)$. This implies that the perturbation satisfies 
\begin{equation}
L_{11}=0, ~~~ p_{00} L_{00} + p_{01} L_{01} + p_{10} L_{10} = 0.
\label{eq:con11}
\end{equation}

For any local maxima of $I(U;Y_{1})+I(V;Y_{2})-I(U;V)$, the derivative with respect to $\epsilon$ must be zero for all perturbations satisfying \eqref{eq:con11}, i.e.
\begin{equation} H_L(U,V) = H_{E(L|U,Y_1)}(U,Y_1) + H_{E(L|V,Y_2)}(V,Y_2). \label{eq:main1} \end{equation}
The terms $H_L(U,V), H_{E(L|U,Y_1)}(U,Y_1), H_{E(L|V,Y_2)}(V,Y_2)$ correspond to
\begin{align*}
H_L(U,V) & = -p_{00}L_{00}\log p_{00}-p_{10}L_{10}\log p_{10} - p_{01}L_{01} \log p_{01} , \\
H_{E(L|U,Y_1)}(U,Y_1) & = - (p_{00}L_{00} + p_{01}L_{01})\log (\frac{p_{00}+p_{01}}{2}) - \frac{p_{10} L_{10}}{2}  \log \frac{p_{10}}{2} \\
&\qquad - \frac{p_{10} L_{10}}{2} \log (\frac{p_{10}}{2} + p_{11}), \\
H_{E(L|V,Y_2)}(V,Y_2) & = - ( p_{00}L_{00} + p_{10}L_{10}) \log (p_{00} + p_{10}) - p_{01} L_{01} \log (p_{01} + \frac{p_{11}}{2}).
\end{align*}

\subsubsection{Case 1.1 $p_{00},p_{01},p_{10},p_{11} > 0$}

In this case the conditions \eqref{eq:con11} and \eqref{eq:main1} imply that the following equalities hold:
\begin{align*}
\frac{p_{00}}{p_{01}} &=  \frac{p_{00}+p_{10}}{p_{01}+\frac{p_{11}}{2}}, \\
\frac{p_{00}}{p_{10}} &=  \frac{p_{00}+p_{01}}{\sqrt{p_{10}(p_{10}+2p_{11})}}.
\end{align*}
These conditions are obtained by setting $L_{10}=0$ and $L_{01}=0$ respectively.

The above two conditions imply that
\begin{align*}
\frac{p_{01}}{p_{00}} = 2, \frac{p_{11}}{p_{10}} = 4.
\end{align*}
These two equalities along with $p_{00}+p_{01}+p_{10}+p_{11}=1$ implies that any non-trivial local maxima is of the form\footnote{This local maxima exists only when $\p(X=1) = p_{11} \leq \frac 45$, and hence there is no local maxima  when $\p(X=1) > \frac 45$. When $\p(X=1) \geq \frac 45$, there was a simple argument in \cite{naw08} that established the conjecture. It is curious that both the  approaches lead to a simple proof in this regime.}
$$ p_{00}=\frac{1-t}{3}, p_{01}=\frac{2(1-t)}{3}, p_{10}=\frac{t}{5}, p_{11}=\frac{4t}{5}.$$

We need to verify the conjecture at this point. It suffices to show that 
\begin{align*}
& I(U;Y_1) + I(V;Y_2) - I(U;V) \leq I(X;Y_1) \left( \leq \max \{I(X;Y_1), I(X;Y_2)\} \right).
\end{align*}
This is equivalent to showing (for $0 \leq t \leq 1$)
\begin{align*}
0 &\leq H(Y_1|U) - H(Y_1|X) - H(V|U) + H(V|Y_2) \\
& = (1-t) + th(\frac{1}{10}) - (1-\frac{4t}{5}) - (1-t)h(\frac 13) - th(\frac 15) + (1-\frac{2t}{5}) h(\frac 13) \\
& = t\left (\frac 45 - 1 + h(\frac{1}{10}) + h(\frac 13) - h(\frac 15) -\frac{2}{5} h(\frac 13)\right) \\
& = \frac{3t}{5} \left( h(\frac 13) - \frac{3}{2} h(\frac 19) \right),
\end{align*}
and this is clearly true as $ \frac{3}{2} h(\frac 19)  \leq h(\frac 13)$. This proves the validity of the conjecture when $p_{00},p_{01},p_{10} > 0$.

\subsubsection{Case 1.2 $p_{01}=0; p_{00},p_{10},p_{11} >0$}

In this case the conditions \eqref{eq:con11} and \eqref{eq:main1} imply that the following equality holds:
$$\frac{p_{00}}{p_{10}} = \frac{p_{00}}{\sqrt{p_{10}(p_{10}+2p_{11})}}. $$
However this cannot hold if $p_{00}, p_{10}, p_{11} > 0$. 

\subsubsection{Case 1.3 $p_{10}=0; p_{00},p_{01},p_{11}>0$}

In this case the conditions \eqref{eq:con11} and \eqref{eq:main1} implies  that
$$\frac{p_{00}}{p_{01}} = \frac{p_{00}}{p_{01}+\frac{p_{11}}{2}}. $$
Again this cannot hold if $p_{00}, p_{10}, p_{11} > 0$. 

\subsubsection{Case 1.4 $p_{00}=0; p_{10},p_{01},p_{11}>0$}

In this case the conditions \eqref{eq:con11} and \eqref{eq:main1} implies 
\begin{equation}
\sqrt{p_{10}(p_{10}+2p_{11})} = p_{01}+\frac{p_{11}}{2}.
\label{eq:conzero}
\end{equation}

To eliminate this possibility, we show that any point that satisfies \eqref{eq:conzero} cannot be a local maxima. Observe that for a local maxima one also requires
$\frac{\partial^2}{\partial \epsilon^2} I(U;Y_{1})+I(V;Y_{2})-I(U;V) \leq 0$, i.e.
$$E[E[L_{UV}|UY_{1}]^2]+ E[E[L_{UV}|VY_{2}]^2]-E[E[L_{UV}|UV]^2] \leq 0.$$
Equivalently for all perturbations satisfying $L_{11}=0$ and $p_{01} L_{01} + p_{10} L_{10} = 0$, any local maxima must satisfy
\begin{align*}
 p_{01} L_{01}^2 + p_{10} L_{10}^2 & \geq p_{01} L_{01}^2 + \frac 12 p_{10} L_{10}^2 + \frac 12 \frac{p_{10}^2}{p_{10}+2p_{11}} L_{10}^2 \\
& \quad + p_{10} L_{10}^2 +  \frac{p_{01}^2}{p_{01}+\frac{p_{11}}{2}} L_{01}^2,
\end{align*}
which is clearly not possible when $p_{10},p_{01},p_{11}>0$. This completes the proof of Case 1.

\subsection{Case 2: $X=U \oplus V$ }

We again prove the conjecture in this case by studying the local maxima. As before, the conjecture is true when two of the four terms $p_{00},p_{01},p_{10},p_{11}$ are identically $0$. When this happens, then the condition reduces to $X=\bar{U},V=1$, $X=\bar{V},U=1$, $X=0$, $X=1$, $X=U,V=0$ or $X=V,U=0$, each of which is solved by Claim \ref{cl:triv}.  Therefore, we only establish the validity of the conjecture for the remaining cases.

Consider a perturbation\footnote{Note that this perturbation is a more general perturbation that the one we have used so far, the multiplicative perturbation of the form
$q(u,v,x)=p(u,v,x)(1+ \epsilon L(u,v,x))$. The multiplication perturbation ensures that if $p(u,v,x)=0$ then $q(u,v,x)=0$; however an additive one 
need not preserve this. Setting $\lambda(u,v,x) = p(u,v,x) L(u,v,x)$ shows that the multiplicative perturbation is a special case of the additive perturbation.
It turns out that in the case $X=U \oplus V$, the analysis of the local maxima is greatly simplified if we consider an additive perturbation; as we are finding the local maxima over a possibly larger space.} $q(u,v,x)=p(u,v,x)+ \epsilon \lambda(u,v,x)$ for some $\epsilon > 0$. For this to be a valid perturbation we require 
\begin{equation}
\lambda_{001}, \lambda_{010}, \lambda_{100}, \lambda_{111} \geq 0
\label{eq:one}
\end{equation} 
as the corresponding $p(u,v,x)$ are zero.
Further let us require that the perturbation maintains $\p(X=0)$. This implies that the perturbation satisfies 
\begin{align}
\lambda_{000} + \lambda_{010} + \lambda_{100} + \lambda_{110} & = 0 \nonumber \\
\lambda_{001} + \lambda_{011} + \lambda_{101} + \lambda_{111} & = 0
\label{eq:con22}
\end{align}

For any perturbation that satisfies \eqref{eq:one} and \eqref{eq:con22} at any local maximum it must be true that the first derivative cannot be positive. This implies
\begin{equation}
H_\lambda(U,V) - H_{E(\lambda|U,Y_1)}(U,Y_1) - H_{E(\lambda|V,Y_2)}(V,Y_2) \leq 0,
\label{eq:consl}
\end{equation}
where
\begin{align*}
H_\lambda(U,V) &= -(\lambda_{001}+ \lambda_{000})\log p_{00} - (\lambda_{010}+\lambda_{011}) \log p_{01} - (\lambda_{100}+\lambda_{101})\log p_{10}\\
&\quad - (\lambda_{110} + \lambda_{111})\log p_{11} \\
H_{E(\lambda|U,Y_1)}(U,Y_1) & = - \frac{(\lambda_{000}+\lambda_{010})}{2}\log (\frac{p_{00}}{2})-\left(\frac{(\lambda_{000}+\lambda_{010})}{2}+\lambda_{001}+\lambda_{011}\right)\log (\frac{p_{00}}{2}+p_{01})\\
& \quad -
\frac{(\lambda_{110}+\lambda_{100})}{2}\log (\frac{p_{11}}{2})-\left(\frac{(\lambda_{110}+\lambda_{100})}{2}+\lambda_{101}+\lambda_{111}\right)\log (\frac{p_{11}}{2}+p_{10}) \\
H_{E(\lambda|V,Y_2)}(V,Y_2) & = -
\left(\frac{(\lambda_{001}+\lambda_{101})}{2}+\lambda_{100}+\lambda_{000}\right)\log (\frac{p_{10}}{2}+p_{00})-\frac{(\lambda_{001}+\lambda_{101})}{2}\log (\frac{p_{10}}{2})\\
&\quad -
\left(\frac{(\lambda_{011}+\lambda_{111})}{2}+\lambda_{010}+\lambda_{110}\right)\log (\frac{p_{01}}{2}+p_{11})-\frac{(\lambda_{011}+\lambda_{111})}{2}\log (\frac{p_{01}}{2}).
\end{align*}

\subsubsection{Case 2.1 $p_{00},p_{01},p_{10},p_{11} > 0$}

Let $a,b,c,d \geq 0$ and let us choose $\lambda_{001}= a = -\lambda_{000}$, $\lambda_{100}= b = -\lambda_{101}$, $\lambda_{010}= c = -\lambda_{011}$, and $\lambda_{111}= d = -\lambda_{110}$.
Observe that this choice satisfies \eqref{eq:one} and \eqref{eq:con22}.  Therefore from the constraint \eqref{eq:consl}, we must have for all choices of $a,b,c,d \geq 0$
\begin{align}
0 & \leq \frac{(a-c)}{2} \log \frac{p_{00}}{p_{00}+2p_{01}} + \frac{(d-b)}{2} \log \frac{p_{11}}{p_{11}+2p_{10}} \nonumber  \\
& \qquad + \frac{(b-a)}{2} \log \frac{p_{10}}{p_{10}+2p_{00}} + \frac{(c-d)}{2} \log \frac{p_{01}}{p_{01}+2p_{11}}
\label{eq:fir}
\end{align}

\begin{itemize}
\item Setting $a=c=k$, $b=d=l$ we require
$$ \frac{(l-k)}{2} \log \frac{p_{10}(p_{01}+2p_{11})}{p_{01}(p_{10}+2p_{00})} \geq 0 $$
for all $l,k \geq 0 $ which is true if and only if
\begin{equation}
p_{10}p_{11} = p_{01}p_{00}.
\label{eq:eq1}
\end{equation}
\item Setting $a=b=l$, $c=d=k$ we require
$$ \frac{(l-k)}{2} \log \frac{p_{00}(p_{11}+2p_{10})}{p_{11}(p_{00}+2p_{01})} \geq 0 $$
for all $l,k \geq 0 $ which is true if and only if
\begin{equation}
p_{01}p_{11} = p_{10}p_{00}.
\label{eq:eq2}
\end{equation}
\item Setting $a=d=l$, $b=c=k$ we require
\begin{equation} \frac{(l-k)}{2} \log \frac{p_{00}}{p_{00}+2p_{01}} \frac{p_{11}}{p_{11}+2p_{10}} \frac{p_{10}+2p_{00}}{p_{10}} \frac{p_{01}+2p_{11}}{p_{01}} \geq 0 
\label{eq:eq3}
\end{equation}
for all $l,k \geq 0$. Observe that equations \eqref{eq:eq1} and \eqref{eq:eq2} imply that $p_{00}=p_{11}$  and $p_{01}=p_{10}$. Let $p=p_{00}=p_{11}, q=p_{01}=p_{10}$.
Substituting this choice into \eqref{eq:eq3} implies that
$$ \frac{(l-k)}{2} \log \left(\frac{p}{p+2q}\right)^2 \left(\frac{q+2p}{q}\right)^2 \geq 0 $$
for all $l,k \geq 0$ which is true if and only if $p=q$.
\end{itemize}

Therefore the only choice of $p_{00}, p_{01}, p_{10}, p_{10}>0$ that satisfies the constraint \eqref{eq:fir} for all choices of $a,b,c,d \geq 0 $ is when
$p_{00}= p_{01} = p_{10} = p_{10} = \frac 14$. In this case observe that $U$, $V$ and $X$ are mutually independent, and the conjecture is trivially true as
$I(U;Y_1) + I(V;Y_2) - I(U;V) = 0$. 

\subsubsection{Case 2.2 One among $p_{00},p_{01},p_{10},p_{11}$ is zero}

All these cases are similar to each other and reduces to a particular $X = U \land V$ case, and hence the validity of the conjecture follows..
For example, when $p_{00}=0$,  observe that $X=0$ {\em if and only if} $U=V=1$. 
Therefore this can also be viewed as a special case of $\bar{X} = U \land V$. (Note that we have already shown the  equivalence between the 
$X = U \land V$ and $X = U \lor V$ cases.)

\medskip

\begin{tabular}{|l|l|}
\hline
Condition & Equivalent $X = U \land V$ case \\
\hline & \vspace*{-0.4cm}\\
 $p_{00}=0 $& $\bar{X} = U \land V$ \\
$p_{01}=0$ & $X = U \land \bar{V}$ \\
$p_{10}=0$ & $X = \bar{U} \land V$ \\
$p_{11}=0$ & $\bar{X} = \bar{U} \land \bar{V}$ \\
\hline
\end{tabular}

\medskip

Since the conjecture was established when $X = U \land V$, this equivalence completes the proof when $X = U \oplus V$.
Thus  Conjecture \ref{conj:intro} is established.

\section{Sum-rate evaluations of inner and outer bounds for BSSC}

We shall evaluate the inner and outer bounds for the BSSC from \cite{nae07} and \cite{naw08}. Apart from completeness, this section serves three purposes:
\begin{itemize}
\item We present a proof that to compute the maximum sum-rate of the Marton's achievable region it suffices to restrict ourselves to $|W| \leq |X|$
\item We correct a minor typo in the evaluation of the maximum sum-rate of the outer bound presented in \cite{nae07}.
\item We also compute the maximum sum-rate obtained via the Korner-Marton outer bound for the BSSC.
\end{itemize}

\subsection{On sum-rate evaluation of Marton's inner bound}

Though this evaluation was done  in \cite{naw08}, assuming the conjectured inequality; we present a slightly different, albeit more general, argument that produces the same result. We first prove that for any broadcast channel it suffices to restrict ourselves to $|W| \leq |X|$ to compute the maximum sum-rate of the Marton's achievable region.  In \cite{naw08} we proved this fact using some properties of the BSSC channel and here we present a general argument.

\begin{claim}
For a discrete memoryless broadcast channel,
to compute the maximum of 
$$ \lambda I(W;Y_1)+ (1-\lambda)I(W;Y_2) + I(U;Y_1|W) +
I(V;Y_2|W) - I(U;V|W), 0 \leq \lambda \leq 1 $$ over all choices of
$(U,V,W) \to X \to (Y_1,Y_2)$
it suffices to restrict to $|\cW| = \cX|$.
\end{claim}

\begin{proof}
Let $p(u,v,w,x)$ achieve a maximum of the above expression. As before, we consider multiplicative Lyapunov
perturbation defined by $q(u,v,w,x)= p(u,v,w,x)(1+ \varepsilon
L(w))$. For $q(u,v,w,x)$ to be a valid probability distribution we
require the conditions $1+\varepsilon L(w) \geq 0 $, $\forall w$ and
$\sum_w p(w)L(w) = 0$. Further let us require that the perturbation
maintains $\p(X=x)$, that is
\begin{equation} \label{eqn:expe}
E(L|X=x) = \sum_{w} p(w|x)L(w) = 0.
\end{equation}
{\it Remark:} There exists nontrivial $L(w)$ if $|\mathcal W|> |\mathcal X|$. 

Observe
that
\begin{align}
&\lambda I_q(W;Y)+ (1-\lambda)I_q(W;Z) + I_q(U;Y|W) +
I_q(V;Z|W) - I_q(U;V|W)  \label{eq:one}\\
&\quad  = +\lambda H_p(Y) + (1-\lambda)H_p(Z) + \lambda \big( H_p(W,Z) +
\varepsilon H^L_p(W,Z) \big) + (1-\lambda) \big( H_p(W,Y) + \varepsilon
H^L_p(W,Y) \big) \nonumber \\
&\qquad - H_p(U,W,Y) - \varepsilon H^L_p(U,W,Y)- H_p(V,W,Z) -
\varepsilon H^L_p(V,W,Z) + H_p(U,V,W) + \varepsilon H^L_p(U,V,W) \nonumber
\end{align}
where
\begin{align*}
H^L_p(W,Y) &= -\sum_{w,y} p(w,y)L(w) \log p(w,y), \\
H^L_p(W,Z) &= -\sum_{w,z} p(w,z)L(w) \log p(w,z), \\
H^L_p(U,V,W) &= -\sum_{u,v,w} p(u,v,w)L(w) \log p(u,v,w), \\
H^L_p(U,W,Y) &= -\sum_{u,w,y} p(u,w,y)L(w) \log p(u,w,y), \\
H^L_p(V,W,Z) &= -\sum_{v,w,z} p(v,w,z)L(w) \log p(v,w,z).
\end{align*}
The first derivative with respect to $\varepsilon$ being zero
implies $$\lambda H^L_p(W,Z) + (1-\lambda) H^L_p(W,Y) -
H^L_p(U,W,Y) - H^L_p(V,W,Z) + H^L_p(U,V,W) = 0.$$ 

Substituting this into \eqref{eq:one} we obtain
\begin{align*} & \lambda I_q(W;Y)+ (1-\lambda)I_q(W;Z)  +
I_q(U;Y|W) + I_q(V;Z|W) - I_q(U;V|W) \\
& \quad = \lambda I_p(W;Y)+
(1-\lambda)I_p(W;Z)  + I_p(U;Y|W) + I_p(V;Z|W) - I_p(U;V|W)
\end{align*}
for any valid perturbation that satisfies (\ref{eqn:expe}).

Now we choose $\varepsilon$ such that $\min_w 1+\varepsilon L(w)
=0$, and let $w=w^*$ achieve this minimum. Observe that $q(w^*) = 0$
and hence there exists an $W$ with cardinality equal to $|\mathcal W|
-1$ such that $\lambda I(W;Y)+ (1-\lambda)I(W;Z)  + I(U;Y|W) +
I(V;Z|W) - I(U;V|W)$ is preserved. We can proceed by induction
until $|\mathcal W| = |\mathcal X|$. This completes the proof of
this claim.
\end{proof}

\subsubsection{Evaluation of the maximum sum-rate of Marton's region for BSSC}

Clearly from the above claim we can assume that $|\cW|=2$. Observe that for the BSSC, $I(X;Y_1) \geq I(X;Y_2)$ if and only if $\p(X=0) \leq \frac 12$. If $0 \leq \p(X=0|W=0) , \p(X=0|W=1) \leq \frac 12$ it is easy to see, using the established inequality that
$$ SR = \min (I(W;Y_1), I(W;Y_2)) + I(X;Y_1|W) \leq I(X;Y_1) \leq C, $$
where $C$ is the single channel capacity given by $h(0.2) - 0.4 \approx 0.321928..$ 
Similarly if $\frac 12 \leq \p(X=0|W=0) , \p(X=0|W=1) \leq 1$ again the sum-rate will be bounded by C. Hence we can assume that
$ 0 \leq \p(X=0|W=0) \leq \frac 12 \leq \p(X=0|W=1) \leq 1$.

Let $d = \max_{p(x)} I(X;Y_1) - I(X;Y_2)$. Then we can solve for $d = 0.10072952..$ and the optimizing choice for $\p(X=0) = 0.15843497..$. Now observe that
\begin{align*} 
SR & \leq I(W;Y_1) + \p(W=0) I(X;Y_1|W=0) + \p(W=1) I(X;Y_2|W=1)\\
& = I(X;Y_1) + \p(W=1) (I(X;Y_2|W=1) - I(X;Y_1|W=1)) \\
& \leq I(X;Y_1) + \p(W=1) d.
\end{align*}
Similarly
\begin{align*} 
SR & \leq I(W;Y_2) + \p(W=0) I(X;Y_1|W=0) + \p(W=1) I(X;Y_2|W=1)\\
& = I(X;Y_2) + \p(W=0) (I(X;Y_1|W=0) - I(X;Y_2|W=0)) \\
& \leq I(X;Y_2) + \p(W=0) d.
\end{align*}

From these two (by adding them) we can deduce that
$$ 2 SR \leq I(X;Y_1) + I(X;Y_2) + d. $$

The maximum of $I(X;Y_1) + I(X;Y_2)=0.6225562..$ occurs when $\p(X=0) = \frac 12$ and hence substituting we obtain that $ SR \leq 0.36164288...$

To show that it is indeed on the boundary of the achievable region consider the joint distribution on $X$ and $W$ as follows:\\
$p(W=0)=p(W=1)=\frac 12$\\
$p(X=0|W=0)=0.15843497.. \mbox{ and } p(X=0|W=1)= 0.84156502..$ \\
For this distribution, all inequalities reduce to equalities and SR of 0.3616.. is achieved.

\subsection{Sum-rate evaluations of the outer bounds for BSSC}

\subsubsection{Case 1: Bound \ref{bd:outer}}
To evaluate maximum of the sum-rate of the outer bound (Bound \ref{bd:outer}) it was shown \cite{nae07} that 
it suffices to consider $\p(X=0) = \frac 12$. (It is immediate using the skew-symmetry of the channel and
the inherent symmetry of the outer bound expressions.)

The sum-rate maximum is hence given by
$$ \max_{p(u,x), \p(x=0)=\frac 12} I(U;Y_1) + I(X;Y_2|U) $$
or in other words maximize
$$ \max_{p(u,x), \p(x=0)=\frac 12} I(X;Y_1) + I(X;Y_2|U) - I(X;Y_1|U)$$

Let $\p(x=0)=x$.
In \cite{naw08} it was shown that the curve $f(x)=I(X;Y_1) - I(X;Y_2)= h(\frac x2) - \frac{1-x}{2} + 1-2x$  is concave when $x \in [0, \frac 12]$ and 
convex when $x \in [0, \frac 12]$. Further it was also shown that the lower convex envelope\footnote{more precisely, in \cite{naw08} the upper concave envelope was characterized, and the characterization of the lower convex envelope follows by symmetry.} was given by
$$ g(x) = \begin{cases} \begin{array}{ll} \frac{5x}{4} f(\frac 45) & 0 \leq x \leq \frac 45 \\ f(x) & \frac 45 \leq x \leq 1 \end{array} \end{cases}. $$

From the definition of the lower convex envelope, we know that
$$   I(X;Y_1|U) - I(X;Y_2|U) \geq g(\frac 12) $$ and it easy to see that the equality is indeed achieved for a binary $U$.

Therefore 
$$ \max_{p(u,x), \p(x=0)=\frac 12} I(X;Y_1) + I(X;Y_2|U) - I(X;Y_1|U) = h(\frac 14) - 0.5 + g(0.5) \approx 0.3725562...$$

This is a correction to the implicit error I made in \cite{nae07} while calculating the lower convex envelope and obtained a bound of $0.37111...$.

\subsubsection{Case 2: Korner-Marton Bound}

To show that this sum-rate is still strictly inside the Korner-Marton\cite{mar79} outer bound observe that we need to evaluate the union over $p(u,x)$
\begin{align*}
R_1 &\leq I(U;Y_1) \\
R_2 & \leq I(X;Y_2) \\
R_1 + R_2 & \leq I(U;Y_1) + I(X;Y_2|U)
\end{align*}

Further if a point $(R_1,R_2)=(a,a)$ belongs to this region, by the skew-symmetry of BSSC, it will also belong to the union over $p(v,x)$
\begin{align*}
R_1 &\leq I(X;Y_1) \\
R_2 & \leq I(V;Y_2) \\
R_1 + R_2 & \leq I(V;Y_1) + I(X;Y_2|V)
\end{align*}
and hence to the intersection of the two regions. The key difference between the bounds is that while the former takes the intersection before the union, the latter takes the union prior to the intersection.

Suppose we wish to compute
$$ \max_{p(u,x)} I(X;Y_1) + I(X;Y_2|U) - I(X;Y_1|U) $$
then from the earlier discussion, this will be the maximum over $x \in [0,1]$ of
$$ h(\frac x2) - x - g(x)$$

It is easy to see that the global maximum will lie when $x \in [0, \frac 45]$ (otherwise maximum occurs when $U$ is trivial and equals $I(X;Y_2)$).
Taking derivatives we obtain that maximum occurs when
$$ \frac 12 \log_2 \frac{2-x}{x} - 1 - \frac 54 f(\frac 45) = 0$$ or 
$$ x^* = \frac{2}{1+2^c} \approx 0.4571429... $$
where $c= 2(1+\frac 54 f(\frac 45)) \approx 1.7548875...$

Thus the maximum sum rate given by
$$ \max_{p(u,x)} I(X;Y_1) + I(X;Y_2|U) - I(X;Y_1|U) \approx 0.3743955... $$

The pair $(U,X)$ that achieves the maximum can be characterized by
$$ \p(U=0)=1-a, \p(X=0|U=0)=0, \p(U=1)=a, \p(X=0|U=0)=\frac 45$$ where
$0.8 * a = x^*$ or $a \approx 0.5714286..$

Observe that for this choice
\begin{align*}
 I(U;Y_1) &= h(\frac{x^*}{2}) - ah(0.4) \approx 0.2206837... \\
 I(X;Y_2|U) & \approx 0.1537118.. \\
 I(X;Y_2) &\approx 0.3006499
 \end{align*}
 
 Therefore the point $(R_1,R_2)=(0.1871978..,0.1871978..)$ lies on the boundary of the Korner-Marton outer bound. In summary, the maximum sum rate given by Korner-Marton outer bound for the BSSC is $0.3743955... $.
 
 \section*{Conclusion}
 We establish an inequality for the binary skew-symmeteric broadcast channel that was conjectured in \cite{naw08}. Thus we have quantified the gap between the outer bounds and the inner bounds for this channel. It would be great to determine which of the bounds are weak, and if possible improve them at least for this interesting channel.

\bibliographystyle{amsplain}
\bibliography{mybiblio}

\end{document}